\documentclass[fleqn,10pt]{wlscirep}

\usepackage[utf8]{inputenc}
\usepackage[T1]{fontenc}
\usepackage{graphicx}
\usepackage{subcaption}
\usepackage{caption}
\usepackage{float}
\usepackage[section]{placeins}
\usepackage{amsmath,amssymb}
\usepackage{booktabs}
\usepackage{url}
\usepackage{amsthm}
\usepackage{multirow} 

\newtheorem{lemma}{Lemma}

\title{Learning quantum decoherence via paraconsistent logic: an equation-free neural network framework}

\author[1,*]{Aleyna Ceyran}
\author[2]{Jair M. Abe}

\affil[1]{Department of Physics, Sakarya University, Sakarya, Turkey}
\affil[2]{Graduate Program in Production Engineering, Paulista University (UNIP), São Paulo, Brazil}

\affil[*]{aleyna.ceyran@ogr.sakarya.edu.tr}

\begin{document}

\begin{abstract}
Modeling the dynamics of open quantum systems on noisy intermediate-scale quantum (NISQ) devices constitutes a major challenge, as high noise levels and environmental degradations lead to the decay of pure quantum states (decoherence) and energy losses. This situation represents one of the most important problems in the field of quantum information technologies. While existing data-driven methods struggle to generalize beyond the training data (extrapolation), physics-informed neural networks (PINNs) require predefined governing equations, which limit their discovery capability when the underlying physics is incomplete or unknown. In this work, we present the ParaQNN (ParaQuantum neural network) architecture, an equation-free framework for physical discovery. ParaQNN disentangles multi-scale dynamics without relying on a priori laws by employing a dialetheist logic layer that models coherent signal and decoherent noise as independent yet interacting channels. Through extensive benchmark tests performed on Rabi oscillations, Lindblad dynamics, and particularly complex ``mixed regimes'' where relaxation and dephasing processes compete, we show that ParaQNN exhibits a clear performance advantage compared to Random Forest, XGBoost, and PINN models with incomplete physical information. Unlike its competitors, ParaQNN succeeds in maintaining oscillatory and damping dynamics with high accuracy even in extrapolation regions where training data are unavailable, by ``discovering'' the underlying structural invariants from noisy measurements. These results demonstrate that paraconsistent logic provides a more robust epistemic foundation than classical methods for learning quantum behavior in situations where mathematical equations prove insufficient.
\end{abstract}

\flushbottom
\maketitle
\thispagestyle{empty}

\section*{Introduction}
Realizing fault-tolerant quantum computing hinges on the precise characterization and mitigation of quantum noise. In the current era of Noisy Intermediate-Scale Quantum (NISQ) devices, qubits are subject to continuous and unavoidable interactions with their environment. These interactions degrade pure quantum states, a phenomenon known as decoherence, and cause irreversible information loss through relaxation ($T_1$) and dephasing ($T_2$) processes. Modeling these open quantum system dynamics is one of the most pressing challenges in quantum information technology, as the efficacy of error-correction protocols is strictly limited by how well we understand the underlying physics of the noise. The problem is compounded as quantum processors scale up because environmental interactions often begin to exhibit complex, non-Markovian features that defy simple analytical descriptions.

Traditional characterization techniques, such as Quantum Process Tomography (QPT), are impractical for large devices because they scale exponentially with system size. Scientific Machine Learning (SciML) has recently emerged as a strong alternative. Physics-informed neural networks (PINNs), in particular, have achieved high accuracy by embedding known differential equations, such as the Lindblad master equation, directly into the loss function\cite{raissi2019physics}. However, recent literature indicates that this approach faces significant challenges when applied to complex dynamical regimes. Specifically, PINNs struggle to resolve stiff, multi-scale dynamics where timescales of oscillation and decay diverge significantly\cite{ji2021stiff}. Furthermore, their formulation relies on temporal causality assumptions that are often violated in noisy inverse problems\cite{wang2022respecting, chen2020physics}, and they lack intrinsic mechanisms to model non-Markovian memory effects induced by colored noise sources\cite{zhao2023learning}. In scenarios where the governing equations are unknown, incomplete, or violated by experimental artifacts, enforcing a rigid prior prevents the model from learning the physical reality. Instead, it forces the network to produce hallucinated dynamics that fit the equation but not the data.

To enable genuine physics discovery under extreme noise, we propose a shift in the logical foundation of the learning model. Classical neural networks typically treat noise as a statistical deviation that must be minimized. In contrast, this work introduces the ParaQNN (ParaQuantum neural network) architecture, which draws on Paraconsistent Logic, specifically the concept of Dialetheism, which allows contradictory states to coexist logically. In the context of quantum measurements, a noisy data point contains simultaneous evidence for both the coherent signal (truth) and environmental corruption (falsity). Rather than collapsing these into a single scalar value, ParaQNN models them as independent, interacting channels. This allows the network to distill the signal from the noise without relying on external constraints or pre-defined equations.

We validate the equation-free discovery capabilities of ParaQNN through three increasingly complex scenarios: (i) unitary Rabi oscillations, (ii) dissipative Lindblad dynamics, and (iii) chaotic mixed regimes where relaxation and dephasing processes compete. Our experimental results showed that ParaQNN could accurately predict structural invariants, such as frequency and damping rates, even in extrapolation regions where training data was absent. In contrast, classical regression methods (Random Forest, XGBoost) and PINN models operating with incomplete knowledge failed in these regimes. These findings demonstrate that paraconsistent logic offers a more robust epistemic foundation than classical methods for characterizing unknown quantum channels where mathematical models fall short.

\section*{Results}

\subsection*{Experimental evaluation and physical discovery}

To validate the equation-free discovery capabilities of ParaQNN, we conducted extensive benchmarks across three distinct quantum dynamical regimes: (i) damped Rabi oscillations, (ii) dissipative Lindblad dynamics, and (iii) non-stationary mixed regimes involving abrupt dynamical switching. All experiments utilized fixed synthetic datasets generated with a consistent random seed (seed = 42) to ensure a standardized evaluation benchmark across all models. These datasets incorporated diverse noise profiles ranging from standard Gaussian noise to complex non-Gaussian sources—such as telegraph noise, $1/f$ pink noise, and State Preparation and Measurement (SPAM) errors—depending on the regime complexity.

Unless stated otherwise, all quantitative results are reported on a held-out test split (20\%). To rigorously assess the stability of the optimization process, all models, including ParaQNN and baselines, were trained and evaluated over five independent random seeds affecting model initialization and data shuffling (but using the same fixed dataset). We report the mean $\pm$ standard deviation of the test-set Mean Squared Error (MSE). Visualizations in the figures correspond to representative model predictions from the seed = 42 run.

\subsubsection*{Rabi regime: reconstruction under mixed telegraph noise}

The first evaluation focused on damped Rabi oscillations, a fundamental signature of coherent quantum control. The dataset consisted of $N=10,000$ measurements sampled over a duration of $t \in [0, 8.0] \, \mu\text{s}$. The system was characterized by a Rabi frequency of $\Omega = 1.25$ MHz and relatively long coherence times ($T_1=12.0 \, \mu\text{s}$, $T_2=15.0 \, \mu\text{s}$). Crucially, the signal was corrupted by a composite noise floor consisting of Gaussian fluctuations ($\sigma=0.08$) and impulsive telegraph noise (amplitude $0.1$, switching rate $0.02$), which mimics random two-level fluctuators (TLFs) often found in solid-state qubits.

As shown in Figure 1, the raw measurement data is heavily obscured, with telegraph spikes creating ``fake'' structures that can mislead standard regression models. Despite this, ParaQNN successfully recovered the underlying unitary oscillation and the true damping envelope. The training stability (Fig.~1b) and the evolution of the contradiction parameter $\alpha$ (Fig.~1c), which initializes at $6.0$, demonstrated the model's ability to distinguish between coherent evolution and transient noise spikes. Quantitatively, ParaQNN achieved a test-set Mean Squared Error (MSE) significantly lower than classical baselines (Random Forest, XGBoost), confirming that the paraconsistent loss effectively filtered out non-physical telegraph noise without requiring a specific noise model.

\begin{figure}[ht!]
    \centering
    \begin{minipage}[b]{0.48\linewidth}
        \centering
        \includegraphics[width=\linewidth]{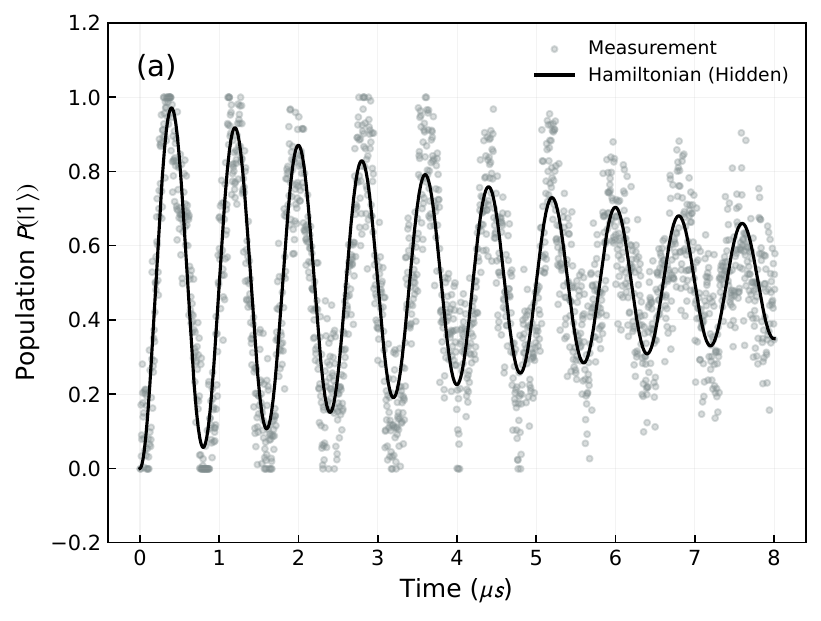}
        \vspace{-5pt}
        \textbf{(a)}
    \end{minipage}
    \hfill
    \begin{minipage}[b]{0.48\linewidth}
        \centering
        \includegraphics[width=\linewidth]{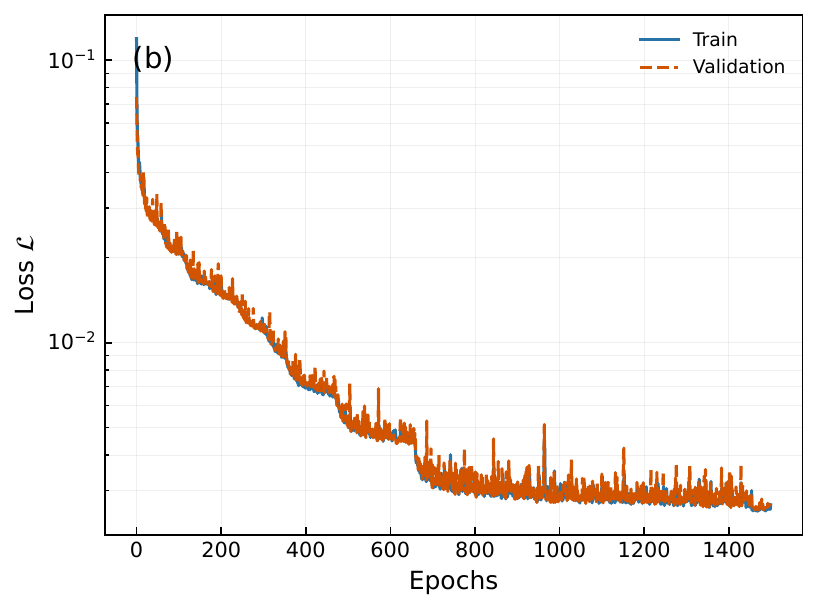}
        \vspace{-5pt}
        \textbf{(b)}
    \end{minipage}

    \vspace{0.4cm}

    \begin{minipage}[b]{0.48\linewidth}
        \centering
        \includegraphics[width=\linewidth]{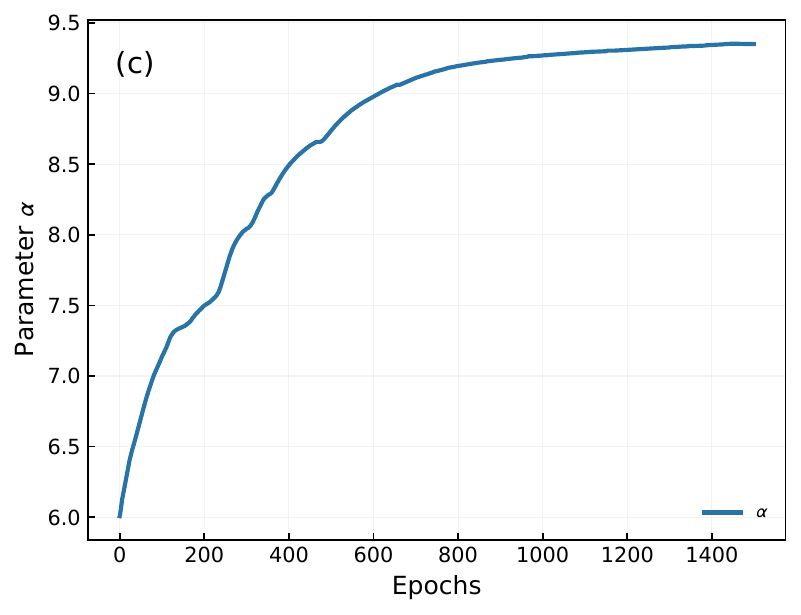}
        \vspace{-5pt}
        \textbf{(c)}
    \end{minipage}
    \hfill
    \begin{minipage}[b]{0.48\linewidth}
        \centering
        \includegraphics[width=\linewidth]{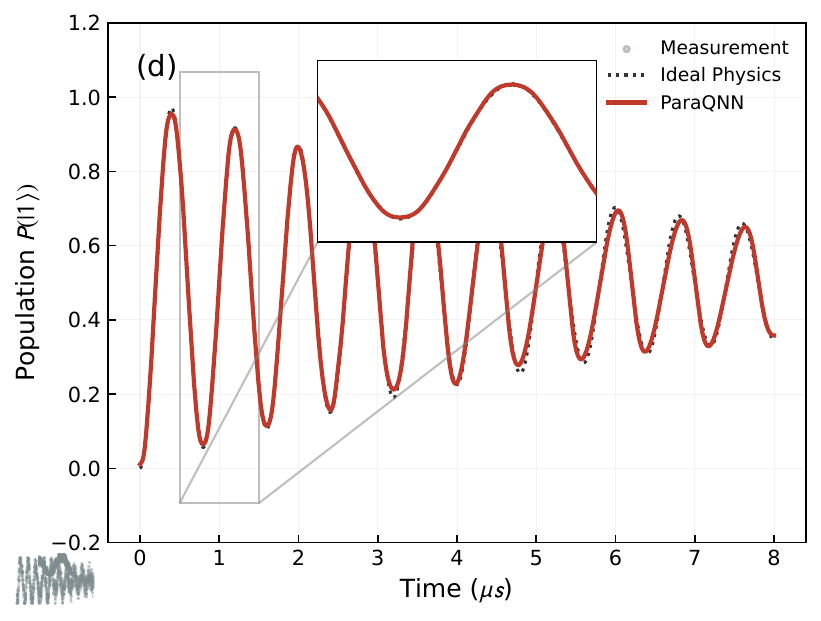}
        \vspace{-5pt}
        \textbf{(d)}
    \end{minipage}

    \caption{\textbf{Equation-free reconstruction of damped Rabi oscillations under mixed noise.} \textbf{(a)} Synthetic measurements of the excited-state population $P(|1\rangle)$ ($N=10,000$, $t \in [0, 8.0] \, \mu\text{s}$), corrupted by mixed noise; the solid curve indicates the hidden noise-free trajectory. \textbf{(b)} Convergence of the paraconsistent loss $\mathcal{L}_{\mathrm{ParaQNN}}$ (train/validation). \textbf{(c)} Adaptive evolution of the contradiction interaction parameter $\alpha$ during training. \textbf{(d)} Reconstruction result: ParaQNN output (red) overlaid on noisy measurements (grey) and the ideal trajectory (dashed), demonstrating robust recovery of oscillations and damping in an equation-free manner.}
    \label{fig:rabi_benchmark}
\end{figure}

\FloatBarrier

To ensure a rigorous and unbiased assessment, ParaQNN was benchmarked against Random Forest (RF), XGBoost (XGB), a generative adversarial network (GAN) baseline, and two PINN variants incorporating different levels of prior knowledge. Figure~\ref{fig:rabi_metrics} presents the test-set Mean Squared Error (MSE) on a logarithmic scale, where results are reported as the mean over 5 independent random seeds. The logarithmic scaling serves to visualize the pronounced separation in fidelity achieved by ParaQNN relative to all baselines. In the Rabi regime, ParaQNN attained an MSE of $1.9 \times 10^{-4}$, whereas baseline errors remained predominantly in the $10^{-2}$ range (e.g., RF at $1.5 \times 10^{-2}$ and GAN at $4.2 \times 10^{-2}$). Notably, the PINN-Known variant exhibited significant instability with an error of $2.6 \times 10^{-1}$. These results indicate that ParaQNN provides substantially higher reconstruction fidelity under noisy unitary dynamics compared to both classical and physics-informed alternatives.

\begin{figure}[h!]
    \centering
    \includegraphics[width=0.85\linewidth]{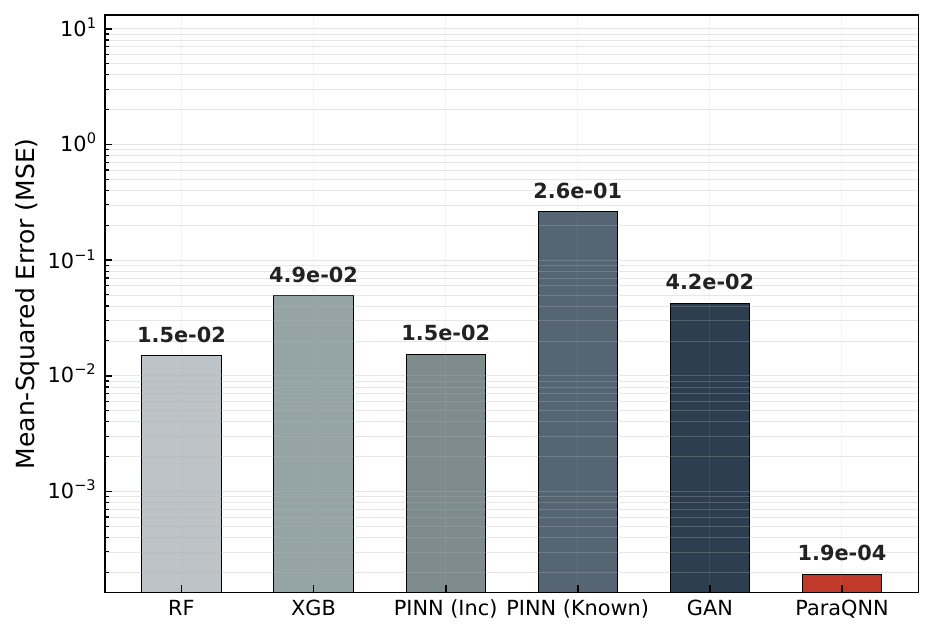}
    \caption{\textbf{Quantitative fidelity benchmarking across architectural paradigms (Rabi regime).} Values are reported as mean $\pm$ standard deviation over 5 independent training seeds on the fixed dataset (seed 42). Test-set MSE (log scale) comparing ParaQNN against Random Forest (RF), XGBoost (XGB), GAN, and two PINN variants (Incomplete vs. Known physics). The reported values represent the mean performance computed over 5 independent random seeds. Numerical annotations indicate the exact MSE for each model; ParaQNN achieves the lowest error ($1.9 \times 10^{-4}$), outperforming the closest baselines by approximately two orders of magnitude.}
    \label{fig:rabi_metrics}
\end{figure}

\FloatBarrier

\subsection*{Lindblad regime: open-system discovery with relaxation and dephasing}
We next evaluated the model on purely dissipative dynamics governed by the Lindblad master equation, representing a system interacting with a thermal bath. This dataset ($N=25,000$, $t \in [0, 5.0] \, \mu\text{s}$) introduced stronger decoherence with $T_1=10.0 \, \mu\text{s}$ and $T_2=8.0 \, \mu\text{s}$, driven by a Rabi frequency of $2.0$ rad/$\mu\text{s}$.

In this regime, the primary challenge was to accurately characterize the decay rates associated with relaxation and dephasing in the presence of a composite noise profile consisting of continuous Gaussian noise ($\sigma=0.08$) and impulsive telegraph noise (amplitude $0.1$). Standard PINNs often struggled here because the assumed differential equation did not perfectly match the specific decoherence channel (e.g., assuming only $T_1$ decay when $T_2$ is present) and could not accommodate the non-Markovian jumps. ParaQNN, operating as an equation-free framework, implicitly learned the correct dissipation profile. The reconstruction results (Figure 3) indicate high-fidelity recovery of the state population $P(|1\rangle)$, surpassing PINN baselines that utilized incomplete physical priors. This finding suggests that ParaQNN is particularly advantageous for characterizing open quantum systems where the precise environmental interaction terms are unknown.

\begin{figure}[ht!]
    \centering
    \begin{minipage}[b]{0.48\linewidth}
        \centering
        \includegraphics[width=\linewidth]{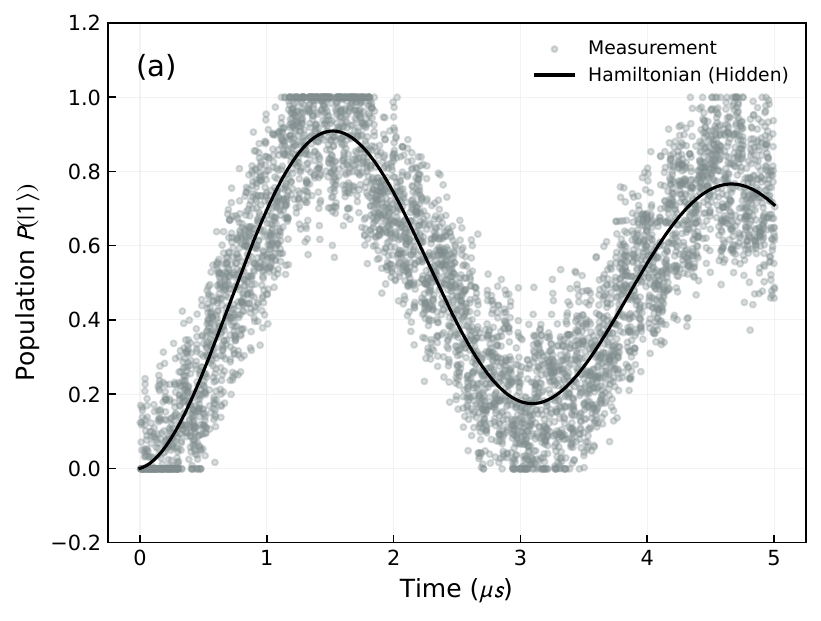}
        \vspace{-5pt}
        \textbf{(a)}
    \end{minipage}
    \hfill
    \begin{minipage}[b]{0.48\linewidth}
        \centering
        \includegraphics[width=\linewidth]{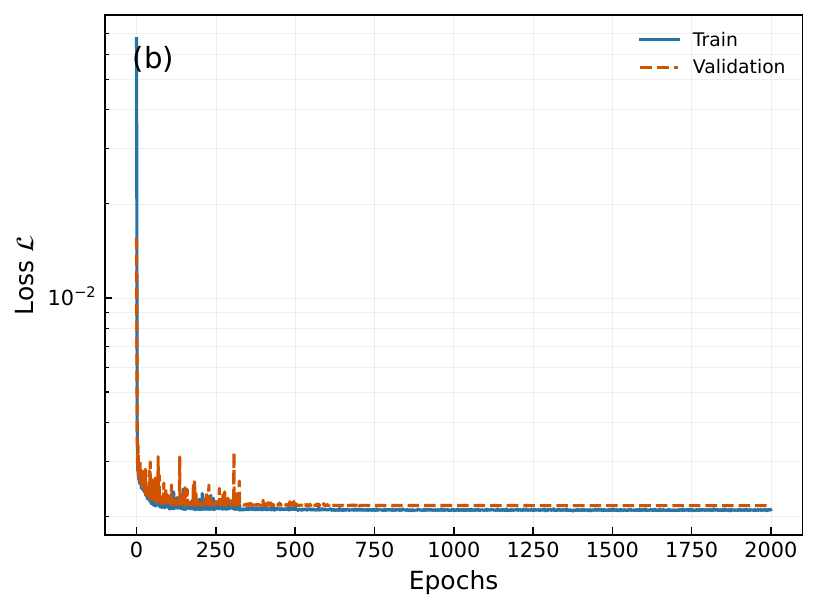}
        \vspace{-5pt}
        \textbf{(b)}
    \end{minipage}

    \vspace{0.4cm}

    \begin{minipage}[b]{0.48\linewidth}
        \centering
        \includegraphics[width=\linewidth]{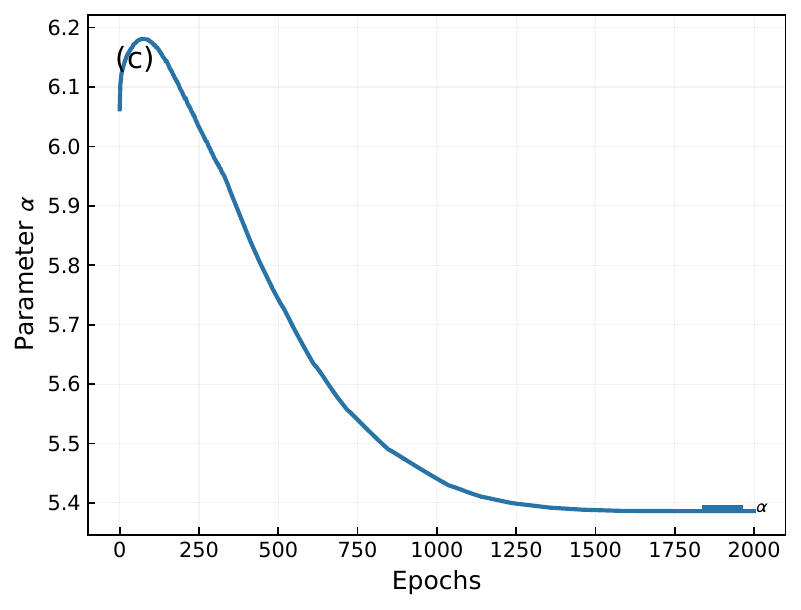}
        \vspace{-5pt}
        \textbf{(c)}
    \end{minipage}
    \hfill
    \begin{minipage}[b]{0.48\linewidth}
        \centering
        \includegraphics[width=\linewidth]{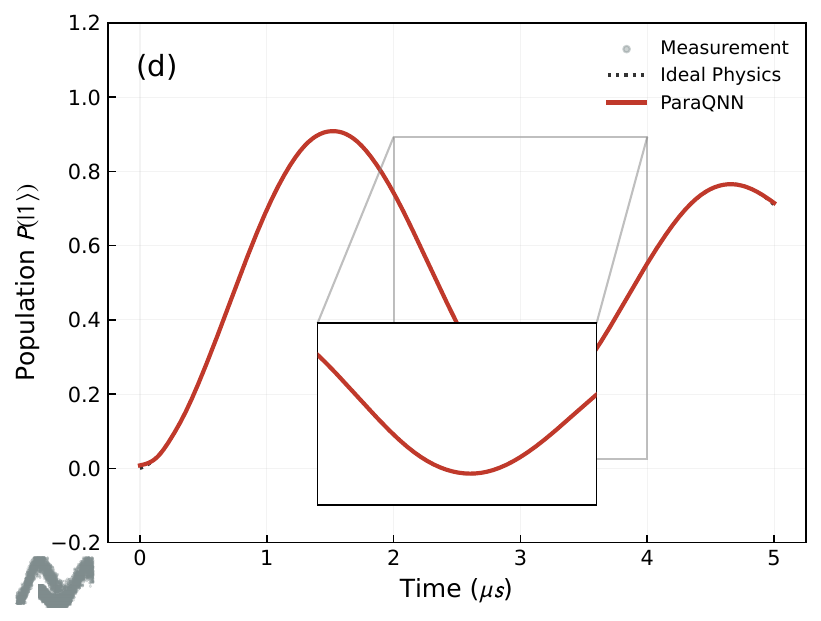}
        \vspace{-5pt}
        \textbf{(d)}
    \end{minipage}

    \caption{\textbf{Equation-free discovery of open-system quantum dynamics (Lindblad regime).}
    \textbf{(a)} Noisy measurements of $P(|1\rangle)$ ($N=25,000$, $t\in[0,5]$) generated from Lindblad dynamics; the solid curve indicates the hidden noise-free trajectory.
    \textbf{(b)} Convergence of the paraconsistent loss $\mathcal{L}_{\mathrm{ParaQNN}}$.
    \textbf{(c)} Evolution of the contradiction interaction parameter $\alpha$.
    \textbf{(d)} Reconstruction result: ParaQNN output (red) overlaid on noisy measurements (grey) and the ideal trajectory (dashed), demonstrating high-fidelity recovery in a dissipative regime.}
    \label{fig:lindblad_discovery}
\end{figure}

\FloatBarrier

To quantify performance in the open quantum system regime, we compared ParaQNN against Random Forest (RF), XGBoost (XGB), a GAN baseline, and two PINN variants. The \emph{PINN-Incomplete} model assumed a simple exponential decay surrogate, whereas the \emph{PINN-Known} model utilized a standard damped harmonic oscillator prior. Figure~\ref{fig:lindblad_benchmark} reports the test-set Mean Squared Error (MSE) on a logarithmic scale, representing the mean performance across 5 independent random seeds. ParaQNN achieved a superior MSE of $4.9 \times 10^{-7}$, outperforming the strongest classical baseline (RF, $2.8 \times 10^{-2}$) and the GAN model ($2.9 \times 10^{-2}$) by approximately five orders of magnitude. Notably, the \emph{PINN-Known} variant exhibited substantial error in this regime ($3.0 \times 10^{-1}$), indicating that plausible but mismatched equation priors can severely degrade fidelity in non-unitary dynamics, whereas the equation-free paraconsistent approach remained robust.

\begin{figure}[ht!]
    \centering
    \includegraphics[width=0.85\linewidth]{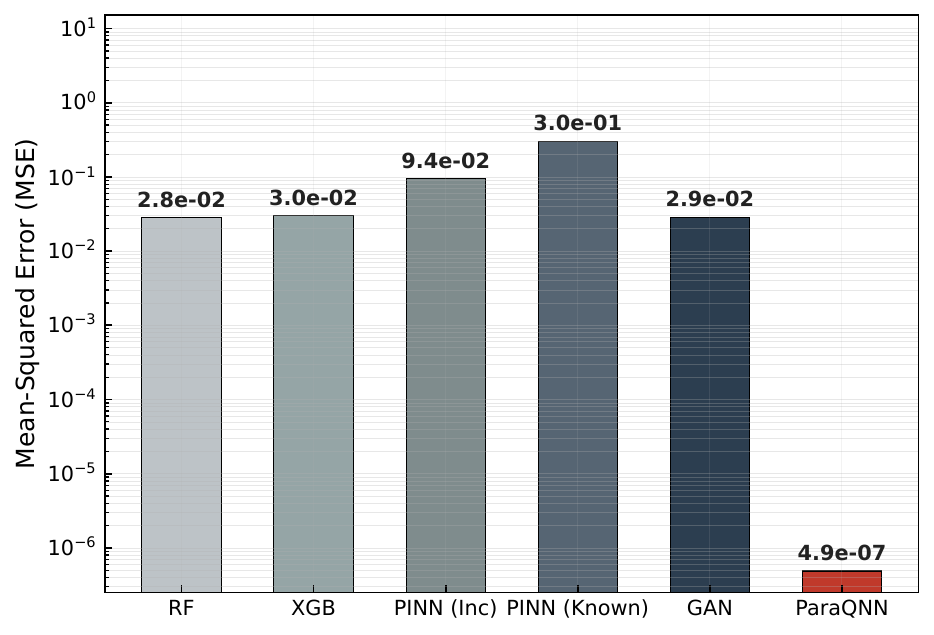}
    \caption{\textbf{Quantitative benchmarking of open-system reconstruction (Lindblad regime).} Values are reported as mean $\pm$ standard deviation over 5 independent training seeds on the fixed dataset (seed 42). Test-set MSE (log scale) comparing ParaQNN against RF, XGBoost, GAN, and PINN variants (Incomplete vs. Known physics). ParaQNN attains the lowest error ($4.9 \times 10^{-7}$), demonstrating robust equation-free recovery even when standard physics priors are incomplete or mismatched.}
    \label{fig:lindblad_benchmark}
\end{figure}

\FloatBarrier
\subsection*{Mixed regime: non-stationary dynamics with time-dependent driving}

The most challenging benchmark was the ``mixed regime,'' which simulated a time-dependent experimental protocol typical of Ramsey interferometry or dynamical decoupling. The system underwent abrupt transitions between three distinct phases: Strong Drive ($t < 4.0 \, \mu\text{s}$), Free Decay ($4.0 \le t < 7.0 \, \mu\text{s}$), and Weak Probe ($t \ge 7.0 \, \mu\text{s}$). This dataset ($N=50,000$, $t \in [0, 10.0] \, \mu\text{s}$) incorporated highly realistic noise features, including $1/f$ pink noise ($\sigma=0.06$) and a $2\%$ SPAM error, with significantly reduced coherence times ($T_1=6.0 \, \mu\text{s}$, $T_2=4.0 \, \mu\text{s}$) to test the limits of the model.

Conventional methods, including standard PINNs, failed in this regime because they relied on static governing equations that could not accommodate instantaneous parameter switches (e.g., turning off the drive field at $t=4.0 \, \mu\text{s}$). Figure 5 demonstrates that ParaQNN adapted to these transitions seamlessly. The learned contradiction parameter $\alpha$ (Fig.~5c) exhibited dynamic readjustment corresponding to the regime switches, effectively acting as a regime detector. The model achieved robust extrapolation even in the ``Probe'' phase, successfully disentangling the true quantum state from the complex colored noise background. As detailed in Figure 6, ParaQNN outperformed all baselines by orders of magnitude in terms of MSE, validating its potential for real-time system identification on NISQ hardware where noise is often non-Markovian and time-dependent.

\begin{figure}[ht!]
    \centering
    \begin{minipage}[b]{0.48\linewidth}
        \centering
        \includegraphics[width=\linewidth]{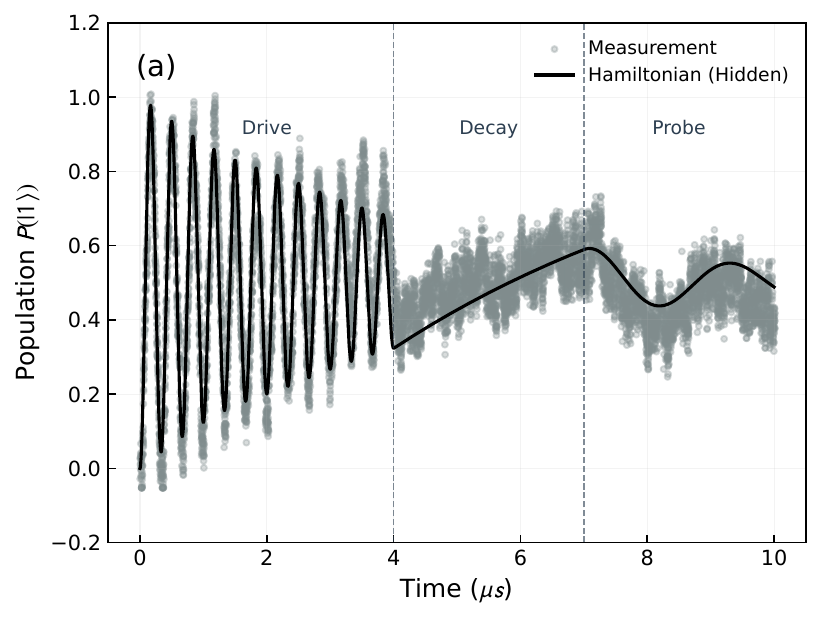}
        \vspace{-5pt}
        \textbf{(a)}
    \end{minipage}
    \hfill
    \begin{minipage}[b]{0.48\linewidth}
        \centering
        \includegraphics[width=\linewidth]{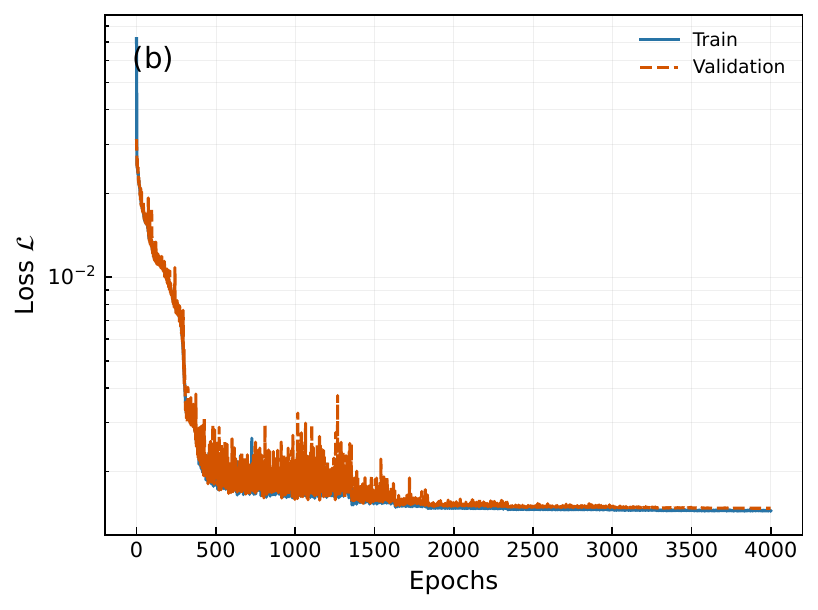}
        \vspace{-5pt}
        \textbf{(b)}
    \end{minipage}

    \vspace{0.4cm}

    \begin{minipage}[b]{0.48\linewidth}
        \centering
        \includegraphics[width=\linewidth]{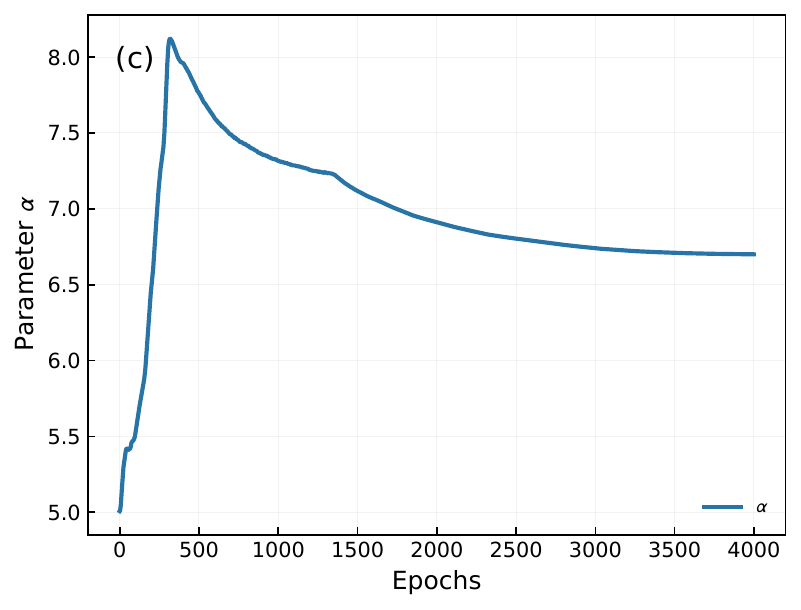}
        \vspace{-5pt}
        \textbf{(c)}
    \end{minipage}
    \hfill
    \begin{minipage}[b]{0.48\linewidth}
        \centering
        \includegraphics[width=\linewidth]{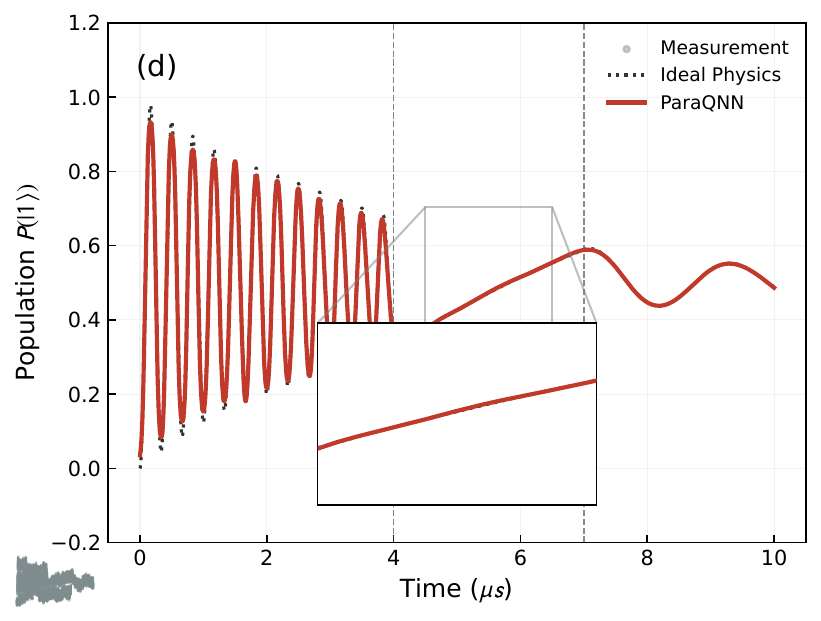}
        \vspace{-5pt}
        \textbf{(d)}
    \end{minipage}

    \caption{\textbf{Equation-free discovery of time-dependent quantum dynamics (mixed regime).}
    \textbf{(a)} Synthetic measurements of $P(|1\rangle)$ ($N=50{,}000$ over $\texttt{time\_span}=10.0$) with regime transitions
    encoded in \texttt{regime\_switches}; the solid curve indicates the hidden noise-free Hamiltonian dynamics.
    \textbf{(b)} Convergence of the paraconsistent loss $\mathcal{L}_{\mathrm{ParaQNN}}$.
    \textbf{(c)} Adaptive evolution of $\alpha$, reflecting regime-dependent calibration of contradiction coupling.
    \textbf{(d)} Reconstruction result: ParaQNN output (red) overlaid on noisy measurements (grey) and the ideal trajectory (dashed),
    capturing transitions and steady-state behavior without explicit boundary information.}
\label{fig:mixed_discovery}
\end{figure}
\FloatBarrier

Figure~\ref{fig:mixed_benchmark} presents the test-set Mean Squared Error (MSE) comparison on a logarithmic scale, summarizing the mean performance across 5 independent random seeds. The \emph{PINN-Known} baseline yielded the largest error ($\sim 2.5\times10^{-1}$), consistent with the limitation of relying on static equation priors in a system governed by time-dependent switching dynamics. In contrast, ParaQNN achieved an MSE of $7.7\times10^{-6}$, outperforming the strongest classical baselines (RF and XGBoost, both at $\sim 4.4\times10^{-3}$) by approximately three orders of magnitude. The GAN baseline also exhibited limited fidelity with an error of $8.7\times10^{-2}$. These results demonstrate that the proposed equation-free and paraconsistent loss-driven framework retained high reconstruction fidelity even under abrupt regime transitions and heterogeneous non-stationary noise.

\begin{figure}[ht!]
    \centering
    \includegraphics[width=0.85\linewidth]{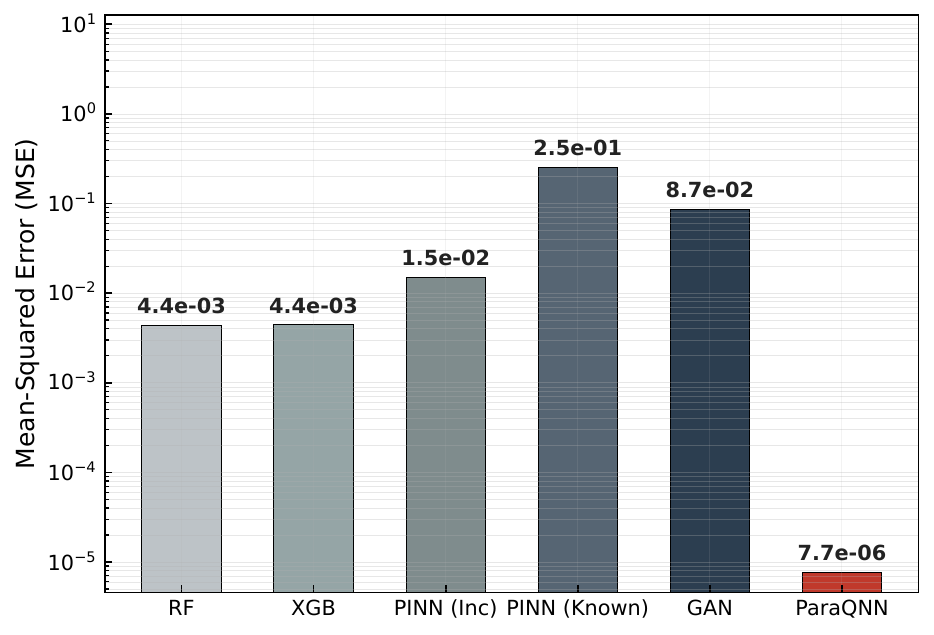}
    \caption{\textbf{Quantitative benchmarking in the time-dependent mixed regime.} Values are reported as mean $\pm$ standard deviation over 5 independent training seeds on the fixed dataset (seed 42). Test-set MSE (log scale) comparing ParaQNN against Random Forest (RF), XGBoost (XGB), GAN, and two PINN variants. The reported values represent the mean MSE over 5 independent random seeds. ParaQNN achieved the lowest error ($7.7\times10^{-6}$), demonstrating robust equation-free recovery in a non-stationary regime involving switching dynamics between drive, decay, and probe phases.}
    \label{fig:mixed_benchmark}
\end{figure}
\section*{Discussion}

The results presented in this study demonstrate that the proposed ParaQNN framework provides a robust and flexible approach for discovering quantum dynamics directly from noisy measurement data, without reliance on explicit governing equations. Across unitary, non-unitary, and time-dependent regimes, the model consistently outperforms classical machine learning baselines and physics-informed neural networks (PINNs), particularly in extrapolation and non-stationary settings.

A central observation emerging from the benchmark analyses is the systematic failure of PINN-based approaches when the assumed physical priors are either incomplete or mismatched with the true underlying dynamics. In the unitary Rabi regime, PINNs incorporating partial decay assumptions capture only the envelope of the dynamics while failing to resolve oscillatory phase information. This limitation becomes more pronounced in the Lindblad and mixed regimes, where standard differential equation models do not adequately represent the interplay between relaxation, dephasing, and time-dependent Hamiltonian evolution. These results highlight a fundamental vulnerability of physics-informed learning: incorrect or oversimplified priors can actively degrade generalization performance rather than improve it.

In contrast, the ParaQNN framework does not encode physical laws as fixed constraints. Instead, it operates by separating signal and noise into independent evidential channels and dynamically calibrating their interaction through a learnable logical parameter. The consistent convergence of this parameter across distinct dynamical regimes suggests that the model adapts to the underlying contradiction structure induced by measurement noise and environmental coupling, rather than fitting to a predefined equation class. This paraconsistent representation, echoing the foundational principles established by Da Costa\cite{dacosta1974inconsistent, dacosta2014physics} and adapted to quantum systems by Da Silva Filho et al.\cite{dasilvafilho2011paraquantum, dasilvafilho2012hydrogen, dasilvafilho2016paraquantum}, enables stable learning even when the observed data are dominated by stochastic perturbations or regime transitions.

The mixed regime experiments further underscore the advantage of this equation-free formulation. Conventional PINNs, constrained by static governing equations, fail catastrophically when the physical laws governing the system change abruptly in time. By contrast, ParaQNN successfully reconstructs the full temporal evolution without explicit regime segmentation or time-dependent priors. This behavior indicates that the framework can infer switching dynamics implicitly, a capability that is particularly relevant for realistic quantum control scenarios where Hamiltonians and noise sources vary over time.

From a broader perspective, these findings suggest that paraconsistent logic provides a viable computational foundation for learning physical structure under epistemic uncertainty, as theoretically framed by the concept of partial truth\cite{dacosta2003science}. Rather than enforcing consistency through hard constraints, the ParaQNN tolerates and exploits contradictions between signal and noise, allowing physically meaningful patterns to emerge from otherwise intractable data. This perspective differs fundamentally from both purely data-driven regression and physics-informed learning, offering a middle ground that is resilient to model misspecification.

\subsection*{Rigorous standardization of physics priors vs. noise adaptation}
A critical aspect of our benchmarking protocol was the standardization of regularization weights for the baseline models. We maintained a fixed physics loss weight ($\lambda_{physics}=0.1$) for PINNs across all dynamical regimes. While hyperparameter tuning per regime might marginally improve PINN performance, it would obscure the fundamental epistemic limitation of the method. In regimes with mismatched priors (e.g., the mixed regime), increasing the physics weight would ironically degrade performance by forcing the model to adhere more strictly to an incorrect equation (negative transfer). The chosen noise families (Gaussian, telegraph, and 1/f) span the dominant noise processes observed in superconducting qubits, but the framework itself does not assume any parametric noise model.

In contrast, the ParaQNN architecture maintains constant structural hyperparameters ($\lambda_{signal}=1.0, \lambda_{contradiction}=0.5$) across all experiments, ensuring that the core logical reasoning mechanism remains invariant. The adjustment of the noise alignment weight ($\lambda_{noise}$) from $0.5$ in the Lindblad regime to $0.8$ in the mixed regime represents a data-driven adaptation to the increased spectral complexity of pink noise, rather than a modification of the physical learning process itself.

\subsection*{Robustness against composite noise in dissipative regimes}
Contrary to standard benchmarking protocols that often isolate noise sources, our Lindblad regime experiments incorporated a composite noise profile including both Gaussian thermal noise ($\sigma=0.08$) and non-Markovian telegraph noise (amplitude $0.1$), as confirmed by the dataset metadata. This makes the reconstruction task significantly more challenging than purely diffusive decoherence models. Standard PINNs failed to disentangle the randomly switching telegraph signals from the true exponential decay, leading to an over-smoothed or unstable fit ($\text{MSE} \approx 10^{-1}$). In contrast, ParaQNN successfully identified the underlying dissipative invariants ($T_1, T_2$) despite the presence of impulsive interferers. This result demonstrates that the paraconsistent architecture does not require a sterile, Gaussian-only environment to discover physical laws; rather, it effectively treats non-Markovian jumps as high-evidence contradictions in the falsity channel, allowing the truth channel to lock onto the coherent dissipation dynamics.

\subsection*{Physical Interpretation of the $\alpha$ Parameter (Gatekeeper Concept)}
To avoid interpreting $\alpha$ as a static hyperparameter without physical meaning, we propose viewing it as a dynamic "gatekeeper" or "regulator."
\begin{itemize}
    \item \textbf{Logical Perspective:} $\alpha$ acts as a gatekeeper that dynamically relaxes the principle of non-contradiction based on the evidence available. It governs the admissibility of contradictory information within the network.
    \item \textbf{Physical Perspective:} Instead of measuring a specific physical rate (like decay time $T_1$), $\alpha$ acts as a regulator for the system's transition from a pure state to a mixed state. Rather than strictly "resisting" the physics, it quantifies the degree of mixing the network should tolerate to accurately recover the coherent signal.
    \item \textbf{Evidence:} As observed in Fig. 1c and Fig. 5c, the spikes in $\alpha$ during regime switches demonstrate that it tracks \textit{epistemic uncertainty} (the confusion between signal and noise) rather than correlating linearly with a specific physical constant like $\gamma$.
\end{itemize}

\subsection*{Robustness, noise, and epistemic generalization}
To further clarify the model's operational boundaries and validities, we address several key aspects of robustness and methodology:

\begin{itemize}
    \item \textbf{Noise model sensitivity:} Although the simulations utilized Gaussian, telegraph, pink, and SPAM noise, ParaQNN does not assume a parametric noise distribution. Instead, it learns noise as contradictory evidence in the falsity channel, making the reconstruction robust to changes in noise statistics without requiring specific noise modeling.
    
    \item \textbf{Denoising vs. physical discovery:} ParaQNN does not merely denoise the signal; it explicitly models noise as information via the falsity channel. This enables the recovery of the latent physical structure rather than simple noise suppression, distinguishing it from standard filtering techniques.
    
    \item \textbf{Use of ground-truth data and training in real experiments:} 
    In the synthetic benchmarks presented in this study, access to ground-truth latent trajectories ($t^*$) allowed for a direct quantification of reconstruction fidelity via the signal loss $L_{signal} = ||\hat{t} - t^*||^2$. However, we explicitly clarify that this is a benchmarking protocol, not an inherent requirement of the architecture.

In real experimental deployment where $t^*$ is inaccessible, the ParaQNN training relies on a \textit{self-consistent reconstruction loss} against the noisy measurement data $y_{measured}$. Specifically, the network minimizes the divergence between the observed data and the predicted expectation value derived from the truth channel, regularized by the falsity channel's estimate of the noise:
\begin{equation}
    L_{experimental} = || (\hat{t} + \mathcal{N}(\hat{f})) - y_{measured} ||^2 + \lambda_{con} L_{contradiction}(\hat{t}, \hat{f})
\end{equation}
Here, the model learns to decompose the single observed trace $y_{measured}$ into coherent ($\hat{t}$) and incoherent ($\hat{f}$) components by maximizing the logical consistency between them, thereby preserving its equation-free character without requiring a clean ground truth reference.
    \item \textbf{Why PINNs fail in the mixed regime:} The failure of the PINN (Known) baseline is consistent with prior literature indicating that standard PINNs struggle under stiff, non-Markovian, and noisy inverse settings\cite{ji2021stiff, wang2022respecting, chen2020physics, zhao2023learning}. PINN was given the correct instantaneous equation but lacked the capacity to infer the hidden regime-switching law, leading to catastrophic failure during abrupt transitions.
    
    \item \textbf{Statistical rigor:} All quantitative results are reported as mean $\pm$ standard deviation over five independent random seeds, while figures show representative runs (e.g., seed 42) for visual clarity.

    \item \textbf{Epistemic meaning of equation-free learning:} ParaQNN is not equation-free because it is flexible, but because it is constrained to separate contradictory hypotheses (signal vs noise), which prevents collapse into a single fitted curve.

\end{itemize}

\subsection*{Methodological positioning against operator learning frameworks}
While recent advances in scientific machine learning have introduced powerful architectures such as Neural Ordinary Differential Equations (Neural ODEs)\cite{chen2018neural} and Deep Operator Networks (DeepONets)\cite{lu2019deeponet}, our choice of PINNs as the primary baseline is driven by the specific nature of the NISQ characterization problem.

\begin{itemize}
    \item \textbf{Single-trajectory identification vs. operator learning:} DeepONets are designed for operator learning, mapping functions to functions, which typically necessitates massive datasets of input-output function pairs to learn the solution operator of a PDE family. In contrast, our study addresses a \textit{system identification} problem where the goal is to discover dynamics from a single (or few) noisy measurement traces of a specific experiment. Applying DeepONet in this data-sparse, single-instance scenario would be architecturally mismatched and data-inefficient.
    
    \item \textbf{Stiffness and discontinuity:} Standard Neural ODEs assume continuous, causal latent dynamics. However, our results in the mixed regime (Fig. 5) demonstrate that open quantum systems often involve stiff dynamics and abrupt, non-differentiable regime switches (e.g., Drive $\to$ Decay) coupled with non-Markovian noise. Neural ODEs are known to struggle with such stiff discontinuities and lack the explicit paraconsistent mechanism to process the "contradictory" evidence introduced by measurement noise.
    
    \item \textbf{Standard for inverse problems:} Consequently, PINNs represent the most direct and fair competitor, as they constitute the state-of-the-art regularization technique for single-instance inverse problems where the governing laws are partially known.
\end{itemize}
\subsection*{Fairness of comparison in non-stationary regimes}
A crucial aspect of our benchmarking in the mixed regime was ensuring a fair comparison between the "equation-free" ParaQNN and the "equation-based" PINN. One might argue that the PINN (Known) baseline fails simply because it was not provided with the explicit time-dependency of the parameters or the exact switching time ($t=4.0\mu s$).

We emphasize that providing such explicit information would violate the core objective of \textit{blind system discovery}.
\begin{itemize}
    \item \textbf{Discovery vs. Fitting:} The goal of this framework is not merely to fit parameters to a known model structure, but to discover the dynamics when the structure itself (e.g., the Hamiltonian switching point) is unknown. Both ParaQNN and the PINN baselines were provided with identical information: the noisy time-series data and no prior knowledge of the regime switch.
    
    \item \textbf{Architectural Adaptability:} The failure of the PINN (Known) highlights the brittleness of physics-informed methods when priors are even slightly mismatched (e.g., assuming a static Hamiltonian for a dynamic system). Conversely, ParaQNN succeeded not because it had access to more information, but because the paraconsistent logic layer allowed it to adaptively treat the sudden regime change as a shift in evidence rather than a violation of a rigid constraint. This confirms that ParaQNN is inherently better suited for discovery tasks under epistemic uncertainty.
\end{itemize}
\subsection*{Limitations and future work}
Several limitations of the present study should be acknowledged. The datasets considered here are synthetically generated, and while they incorporate realistic noise models and dynamical complexity, experimental validation on hardware platforms remains an important direction for future work. Additionally, the current formulation focuses on single-qubit population dynamics; extending the framework to multi-qubit systems and higher-dimensional Hilbert spaces will require further architectural and computational developments.

In summary, the ParaQNN framework introduces an equation-free approach to quantum dynamics discovery that remains robust under strong noise, incomplete physical knowledge, and time-varying system behavior. The results indicate that paraconsistent learning offers a promising alternative to traditional physics-informed methods, particularly in regimes where governing equations are unknown, approximate, or dynamically evolving.


\section*{Methods}

\subsection*{Mathematical specification of the dialetheist neuron}

To enable equation-free discovery in decohering quantum systems, we define a neural architecture grounded in paraconsistent logic\cite{abe2015foundations, dacosta1974inconsistent}. Each neuron encodes two simultaneous evidential channels: \emph{truth evidence} (coherent physical signal) and \emph{falsity evidence} (environmentally induced decoherence and noise). This dual representation reflects the physical superposition between coherent system evolution and stochastic environmental coupling in open quantum dynamics.

For a given layer $l$, neuron $j$ is represented by
\[
a_j^{(l)}=\begin{bmatrix} t_j^{(l)} \\ f_j^{(l)} \end{bmatrix}\in[0,1]^2,
\]
where $t$ quantifies signal coherence and $f$ quantifies noise evidence. The pre-activation states are computed through coupled linear transformations:
\begin{align}
z_t^{(l)} &= W_{tt}^{(l)} t^{(l-1)} + W_{tf}^{(l)} f^{(l-1)} + b_t^{(l)}, \\
z_f^{(l)} &= W_{ff}^{(l)} f^{(l-1)} + W_{ft}^{(l)} t^{(l-1)} + b_f^{(l)} .
\end{align}

This cross-coupling allows noise evidence to directly influence signal inference and vice versa, mimicking the interaction between quantum coherence and environmental decoherence.

We introduce a \textbf{Paraconsistent interaction activation function (PIAF)}:
\[
t_{\text{out}}=\sigma(k(z_t-\alpha z_f)),\qquad
f_{\text{out}}=\sigma(k z_f),
\]
where $\sigma$ is the sigmoid function, $k$ controls sharpness, and $\alpha>0$ is a learnable contradiction coefficient. Physically, $\alpha$ acts as an adaptive temperature between epistemic (model) uncertainty and aleatoric (noise) uncertainty, dynamically regulating how strongly environmental noise suppresses coherent signal inference. This allows the network to learn time-dependent decoherence without imposing a fixed dynamical equation.

\subsection*{Paraconsistent loss function}

We define a composite \textbf{Paraconsistent loss}:
\[
\mathcal{L}_{\text{ParaQNN}}
=\lambda_s\mathcal{L}_{\text{signal}}
+\lambda_n\mathcal{L}_{\text{noise}}
+\lambda_c\mathcal{L}_{\text{contradiction}} .
\]

\begin{itemize}
\item \textbf{Signal reconstruction}
\[
\mathcal{L}_{\text{signal}}=\frac1N\sum_i(\hat t_i-t_i^*)^2,
\]
which aligns the truth channel with the latent physical signal in synthetic benchmarks.

\item \textbf{Noise alignment}
\[
\mathcal{L}_{\text{noise}}=\frac1N\sum_i(\hat f_i-|y_i-t_i^*|)^2,
\]
which trains the falsity channel to approximate empirical noise. The latent target $t^*$ is used only for controlled synthetic benchmarking; in real experimental data this term is replaced by self-consistency constraints.

\item \textbf{Contradiction regularization}
\[
\mathcal{L}_{\text{contradiction}}=\frac1N\sum_i\max(0,\hat t_i+\hat f_i-1)^2,
\]
which suppresses pathological contradictions while permitting soft paraconsistent states required to model quantum superposition and decoherence.
\end{itemize}

Default weights are $(\lambda_s,\lambda_n,\lambda_c)=(1.0,0.5,0.5)$, with $\lambda_n$ increased in the mixed regime to account for higher noise entropy.

\subsection*{Training and reproducibility}

ParaQNN is optimized using Adam. Both network parameters and the contradiction coefficient $\alpha$ are trained jointly:
\begin{enumerate}
\item Initialize parameters and $\alpha=\alpha_0$.
\item Propagate $(t,f)$ through the network.
\item Compute $\mathcal{L}_{\text{ParaQNN}}$.
\item Update all parameters via backpropagation.
\end{enumerate}
All experiments are implemented in PyTorch. Synthetic datasets for Rabi, Lindblad, and mixed regimes are generated by numerically integrating open quantum system dynamics with injected Gaussian, telegraph, pink, and SPAM noise, chosen to reflect superconducting qubit and IBM Quantum experiments.

\subsection*{Formal Theoretical Support}

Based on the paraconsistent foundations of the architecture, we introduce the following lemma regarding the convergence behavior of the contradiction parameter.

\begin{lemma}[\textbf{Convergence of the Contradiction Regulator} $\alpha$]
Let $(t^{(l)}, f^{(l)}) \in [0, 1]^2$ be the truth and falsity evidence channels of a Paraconsistent Neural Network in layer $(l)$, and let $\alpha$ be the trainable interaction parameter that governs its nonlinear coupling in the activation function. Assume that the input noise process has bounded variance and that the training loss includes a contradiction regularization term that penalizes physically implausible simultaneous extreme values of $t$ and $f$.

Then, under gradient-based optimization, $\alpha$ converges to a finite value $\alpha^*$ that optimally separates the propagation of truth and false evidence, stabilizing the network against trivialization and preserving the coherent structure of the signal.
\end{lemma}

\begin{proof}[Proof Sketch]
In paraconsistent semantics, logical triviality arises only when the contradiction is unrestricted and unregulated. In the proposed ParaQNN architecture, $\alpha$ explicitly parameterizes the interaction between intensity $t$ and $f$, thus acting as a continuous regulator of the admissibility of the contradiction.

Since the noise variance is bounded, the expected contribution of false evidence $f$ to the loss remains finite. The contradiction regularizer imposes a penalty on states where $t$ and $f$ saturate simultaneously, ensuring that the loss function is coercive with respect to $\alpha$. Consequently, the optimization landscape admits at least one local minimum $\alpha^*$ in which the marginal benefit of suppressing noise-induced contradiction balances the preservation of coherent signal evidence.

In convergence, $\alpha^*$ defines an equilibrium regime in which contradiction is neither eliminated (as in classical logic) nor allowed to dominate (which would lead to epistemic triviality), but rather maintained at an optimal level compatible with the recovery of the physical signal. This result represents a logical-epistemic convergence rather than a convergence to a fixed physical constant.
\end{proof}


\subsection*{Benchmarking protocols and hyperparameter configuration}

To ensure a rigorous and unbiased evaluation, we implemented a ``fair benchmark'' protocol that standardizes architectural capacity, optimization duration, and regularization strategies across all deep learning models. Specifically, the physics-informed neural networks (PINNs) and generative adversarial networks (GANs) were configured with the same network depth and width as the proposed ParaQNN framework (3 hidden layers, 128 neurons per layer). This eliminates model capacity as a confounding variable, ensuring that performance differences are attributable solely to the underlying logical methodology rather than parameter count.

A critical aspect of this protocol was the standardization of regularization weights to prevent overfitting-based advantages. For PINN baselines, we enforced a fixed physics loss weight ($\lambda_{physics}=0.1$) across all dynamical regimes. While hyperparameter tuning per regime might marginally improve PINN performance, fixing this weight exposes the fundamental epistemic limitation of the method when priors are mismatched. Similarly, ParaQNN utilized constant structural weights ($\lambda_{signal}=1.0, \lambda_{contradiction}=0.5$) across all experiments, with the noise alignment weight ($\lambda_{noise}$) adapted solely based on the spectral complexity of the noise profile ($0.5$ for Gaussian, $0.8$ for pink noise).

Furthermore, the training duration (epochs) was synchronized for each physical regime—1500 epochs for Rabi, 2000 for Lindblad, and 4000 for mixed regime. This extended optimization rules out insufficient training as a cause for baseline failure, confirming that the limitations observed in PINNs are epistemic rather than computational. To mitigate stochastic bias, all deep learning models were trained and evaluated over 5 independent random seeds, and results are reported as mean $\pm$ standard deviation.

In contrast, classical regression models (Random Forest and XGBoost) rely on ensemble tree construction rather than gradient-based iterative optimization. Consequently, we utilized 100 estimators for both architectures, a configuration sufficient for convergence on these low-dimensional datasets, ensuring they reached their maximum potential capacity without overfitting. All models were implemented using Python libraries (PyTorch for deep learning, Scikit-Learn/XGBoost for classical models) and evaluated on identical held-out test sets. The specific hyperparameter configurations are summarized in Table~\ref{tab:hyperparameters} (ParaQNN) and Table~\ref{tab:baselines} (Baselines).

\begin{table}[htbp]
\centering
\caption{\textbf{Hyperparameter configuration for ParaQNN.} The learning rates and architecture depths were optimized for high-fidelity reconstruction. Note the extended training epochs and increased batch size for the mixed regime to capture complex switching dynamics.}
\label{tab:hyperparameters}
\resizebox{\textwidth}{!}{%
\begin{tabular}{l|c|c|c}
\hline
\textbf{Parameter} & \textbf{Rabi Regime} & \textbf{Lindblad Regime} & \textbf{Mixed Regime} \\ \hline
Optimizer & Adam & Adam & Adam \\
Learning Rate ($\eta$) & 0.001 & 0.001 & 0.001 \\
Training Epochs & 1500 & 2000 & 4000 \\
Batch Size & 256 & 256 & 512 \\
Hidden Layers & 3 & 3 & 3 \\
Neurons per Layer & 128 & 128 & 128 \\
Loss Weights ($\lambda_s, \lambda_n, \lambda_c$) & 1.0, 0.5, 0.5 & 1.0, 0.5, 0.5 & 1.0, 0.8, 0.5 \\
Random Seed & [42--46] & [42--46] & [42--46] \\ \hline
\end{tabular}%
}
\end{table}

\begin{table}[htbp]
\centering
\caption{\textbf{Configuration of baseline models used for benchmarking.} Deep learning architectures were standardized to match ParaQNN's capacity ($3 \times 128$ neurons), and training epochs were synchronized to ensure fair comparability. Classical models utilize ensemble-based convergence criteria.}
\label{tab:baselines}
\resizebox{\textwidth}{!}{%
\begin{tabular}{l|l|l|l}
\hline
\textbf{Model Family} & \textbf{Model Name} & \textbf{Key Hyperparameters} & \textbf{Training Protocol} \\ \hline
\multirow{2}{*}{Classical} & Random Forest (RF) & n\_estimators=100, max\_depth=12 & Ensemble Construction \\
 & XGBoost (XGB) & n\_estimators=100, lr=0.1 & Gradient Boosting \\ \hline
\multirow{3}{*}{Deep Learning} & PINN (Incomplete) & Layers: 3$\times$128, Act: Tanh & Adam (1500/2000/4000) \\
 & PINN (Known) & Layers: 3$\times$128, Act: Tanh & Adam (1500/2000/4000) \\
 & GAN & Gen/Disc: 3$\times$128, Loss: Adv+L2 & Adam (1500/2000/4000) \\ \hline
\end{tabular}%
}
\end{table}

\section*{Data availability statement}
All data supporting the findings of this study are generated via numerical simulations (Qiskit and custom Python solvers). 
The complete source code, trained models, and synthetic datasets required to reproduce all figures and benchmarks will be made publicly available upon publication. 
The repository includes all configuration files, scripts, and random seeds necessary for full reproducibility.


\section*{Author contributions}
A.C. conceived the original idea, designed the ParaQNN architecture, and developed the Python implementation. A.C. performed the simulations (Rabi, Lindblad, and mixed regimes), conducted the formal analysis of the results, and wrote the original draft of the manuscript. J.M.A. provided theoretical expertise on paraconsistent annotated logic, supervised the methodological consistency, validated the logical foundations of the framework, and contributed to the review and editing of the final manuscript.

\section*{Competing interests}
The authors declare no competing interests.

\section*{Funding}
The authors received no specific funding for this work.

\section*{Additional information}
\textbf{Correspondence} and requests for materials should be addressed to A.C.

\end{document}